\documentclass[a4paper,10pt]{article}
\usepackage{a4wide}
\usepackage[normalem]{ulem}
\usepackage{dsfont}
\usepackage[T1]{fontenc}
\usepackage{amsmath}
\usepackage{amssymb}
\usepackage{amsthm}
\usepackage[ansinew]{inputenc}
\usepackage{enumerate}
\usepackage{fancyhdr}
\usepackage{color}
\usepackage{ifpdf}
\usepackage{rotating}
\usepackage[T2A]{fontenc}
\usepackage{mathtext}
\DeclareMathOperator{\tr}{tr}

\DeclareMathOperator{\pr}{pr}
 \newcommand{\R}        {\mathds{R}}

\newtheorem{deff}{Definition} 
\newtheorem{thm}{Theorem}
\newtheorem{lem}{Lemma}

\usepackage{authblk}

\begin{document}
\title{Future global non-linear stability of surface symmetric solutions of the Einstein-Vlasov system with a cosmological constant}
\author[1,2]{Ernesto Nungesser\footnote{ernesto@maths.tcd.ie}}
\affil[1]{School of Mathematics, Trinity College, Dublin 2, Ireland}
\affil[2]{Department of Mathematics, Royal Institute of Technology, 10044~Stockholm, Sweden}\maketitle

\begin{abstract}
We show future global non-linear stability of surface symmetric solutions of the Einstein-Vlasov system with a positive cosmological constant.
Estimates of higher derivatives of the metric and the matter terms are obtained using an inductive argument.
In a recent research monograph Ringstr\"{o}m shows future non-linear stability of (not necessarily symmetric) solutions of the Einstein-Vlasov system with 
a non-linear scalar field if certain local estimates on the geometry and the matter terms are fulfilled. We show that these assumptions are satisfied 
at late times for the case under consideration here which together with Cauchy stability leads to our main conclusion.
\end{abstract}
\section{Introduction}
In the present paper we deal with the surface symmetric Einstein-Vlasov system with a positive cosmological constant $\Lambda$. Note that for a negative cosmological constant 
no solution exists cf. Proposition 4.1 in \cite{Tchap2}. We refer to \cite{Rein} for the case of a vanishing cosmological contant. For an introduction to the terminology of surface symmetric spacetimes and some simple special solutions see \cite{rendallcrush}. We will use areal coordinates, i.e. the coordinates are chosen such that the time coordinate equals the so called area radius of a surface of symmetry. The circumstances under which coordinates of this type exist were established in \cite{ARR} for the case of vanishing cosmological constant and extended to the case of positive cosmological constant in \cite{Tchap2,Tchap1}. Future geodesic completeness has been proven for all solutions
with hyperbolic and plane symmetry and also information about the asymptotic behaviour at late times has been obtained cf. Theorems 3.2 and 3.3 in \cite{Tchap1}.
The proof is built on the local existence theorem and continuation criterion previously proved in \cite{Tchap2}. In the case of spherical symmetry, if $t_0>\Lambda^{-\frac12}$
where $t_0$ has to be understood as the initial area radius, future geodesic completeness and information about the asymptotic behaviour at late times could also be shown,  
cf. Theorem 5.2 of \cite{Tchap2}. In that paper an estimate on the decay of the support of the distribution function
and estimates on the matter terms where also shown,
cf. Lemma 6.1 and Proposition 6.2 of \cite{Tchap2}. The differences between the various cases are related to the fact that while an initially expanding solution with hyperbolic
or plane symmetry can be expected to expand forever, a spherically symmetric solution can be expected to recollapse. 
Another ingredient in our paper is a certain combination of terms concerning the characteristic system found in \cite{Rein}, which we are able to control. 
These results form the basis we start with. 

Estimates of higher derivatives of the metric and the matter terms are obtained using an inductive argument. 
In Theorem 7.16  of \cite{Hans} future non-linear stability of (not necessarily symmetric) solutions of the Einstein-Vlasov system 
with a non-linear scalar field is shown, if certain local estimates on the geometry and the matter terms are fulfilled.
The analysis of higher order terms is necessary because the goal is to show that for late enough $t$, 
there is a neighbourhood of each point, such that Theorem 7.16 of \cite{Hans} applies in the neighbourhood which together with Cauchy stability 
yields the desired result. Estimates of higher order are required due to the choice of the spaces of initial data and the associated distance
concepts. This is related to the fact that the solution to the initial value problem of wave equations does not depend continuously on the initial data with respect to,
e.g. the $C^k$-norm (see discussion in Section 2.2 of \cite{Hans}).

The paper is organized as follows. In the next section we briefly introduce the initial value formulation of the Einstein-Vlasov system with a cosmological constant and introduce
some function spaces. For an introduction to the Einstein-Vlasov system we refer to \cite{Hans}. In particular we have taken several definitions from Chapter 7 of \cite{Hans} and
we refer the reader interested in details to that monograph. In Section 3 we introduce the basic equations of the Einstein-Vlasov-system with a cosmological constant and surface
symmetry and collect previous results concerning the asymptotic behaviour of surface symmetric solutions of this system. In the following sections we have benefited from the analysis
presented in \cite{HakanHans} where $\mathbb{T}^2$-symmetric solutions have been analyzed (see \cite{PreHans} for a recent development).
In Section 4 we obtain estimates of higher order which we use in Section 5 to establish some
geometric consequences and the cosmic no-hair theorem. Consequences for the distribution function and our stability result are presented in Sections 6 and 7.

\section{The Einstein-Vlasov system with a cosmological constant}
A cosmological model represents a universe at a certain averaging scale. It is described via a Lorentzian metric $g$
 (we will use signature -- + + +) on a manifold $M$ and a family of fundamental observers. The interaction between the geometry and the matter is described by 
 the Einstein field equations (we use geometrized units i.e., the gravitational constant and the speed of light in vacuum $c$ are set equal to one):
\begin{eqnarray}\label{Einstein}
G+\Lambda g = 8\pi T,
\end{eqnarray}
where $G$ is the Einstein tensor, $T$ is the energy-momentum tensor and $\Lambda>0$ is the cosmological constant. The Einstein summation convention
 that repeated indices are to be summed over is used. Latin indices run from one to three and Greek ones from zero to three.

For the matter model we will take the point of view of kinetic theory. We will consider particles with unit rest mass which move under the influence of the gravitational field. 
The distribution function which gives a statistical description of an ensemble of collections of particles is described by a non-negative real valued function 
$f: P \rightarrow [0,\infty)$ on the mass shell $P$, defined to be the set of future directed unit time like vectors. 
The component $p^0$ of an element $p\in{P}$ can be expressed in terms of the other components due to the mass shell relation:
\begin{eqnarray*}
 p_{\alpha}p^{\alpha}=-1,
\end{eqnarray*}
and will be thus considered as a dependent variable. The relation between $f$, the distribution function on the mass shell, and $\bar{f}$, the distribution function on the initial hypersurface $\Sigma$ induced by $f$, 
is given by the following map:
\begin{eqnarray*}
 \bar{f}=f \circ \pr_{\Sigma}^{-1},
\end{eqnarray*}
where $\pr_{\Sigma}$ is the diffeomorphism from the mass shell above the hypersurface $\Sigma$ to its tangent space $T \Sigma$, defined by mapping a vector to its component 
perpendicular to the normal of $\Sigma$.  
A free particle travels along a geodesic. The equations of motion define a flow on $P$ which is generated
by a vector field $ \mathcal{L}$ which is called geodesic spray or Liouville operator. 
We will consider the collisionless case which is described via the Vlasov equation:
\begin{align}\label{Vlasov}
 \mathcal{L}(f)=0,
\end{align}
where $\mathcal{L}$ is the restriction of the Liouville operator to the mass shell. It remains to define the energy-momentum tensor $T$ in terms of the distribution function
and the metric.   Define the energy-momentum tensor as follows:
 \begin{eqnarray}\label{em}
T_{\alpha\beta}(\zeta)=\int_{P_{\zeta}} f p_{\alpha}p_{\beta}\varpi_{\zeta},
\end{eqnarray}
where $P_{\zeta}$ denotes the mass shell at the spacetime point $\zeta$, $p_{\alpha}$ denotes the components of the one-form obtained by lowering
the index of $p\in{P_{\zeta}}$ with $g$ and $\varpi_{\zeta}$ is the induced (Riemannian) volume of the mass shell considered as a hypersurface
in the tangent space at $\zeta$. If $f$ satisfies the Vlasov equation (\ref{Vlasov}), then 
the energy momentum tensor (\ref{em}) is divergence-free and thus it is compatible with the Einstein equations (\ref{Einstein}). 
We will assume that $f$ is compactly supported in momentum space.
Given $(M,g)$, we say that $(\mathrm{x},U)$ are \emph{canonical local coordinates} if 
$\partial_{{\mathrm{x}}^0}$ is future directed timelike on $U$ and 
$g (\partial_{{\mathrm{x}}^i}\vert_{\zeta},\partial_{{\mathrm{x}}^j}\vert_{\zeta})$ are the components of a positive definite matrix for all $\zeta \in{U}$. 
Given canonical local coordinates and $p\in{P_{\zeta}}$ for some $\zeta\in{U}$ we define local coordinates $\Xi_{\mathrm{x}}$ on the mass shell by the condition
that
 \begin{eqnarray*}
  \Xi_{\mathrm{x}}(p)=\Xi_{\mathrm{x}}(p^{\alpha}\partial_{\mathrm{x}}^{\alpha}\vert_{\zeta})=(\mathrm{x}(\zeta),\bar{p}),
 \end{eqnarray*}
where $\bar{p}=(p^1,p^2,p^3)$. If $f$ is defined on the mass shell, we will use the notation $\mathrm{f}_{\mathrm{x}}= f \circ \Xi_{\mathrm{x}}^{-1}$. Let us introduce the following
function space (Definition 7.1 of \cite{Hans}):
\begin{deff}
 Let $\frac52 < z\in{\mathds{R}}$, $(M,g)$ be a time oriented $3+1$-dimensional Lorentz manifold and $P$ the associated mass shell. The space $\mathcal{F}^{\infty}_{z}(P)$ is defined
 to consist of the smooth functions  $f: P \rightarrow \mathds{R}$, such that for every choice of canonical local coordinates $(\mathrm{x},U)$, $3+1$-multiindex $\alpha$,
 $3$-multiindex $b$, the derivative $\partial_{\mathrm{x}}^{\alpha}\partial_{\bar{p}}^{b} \mathrm{f}_{\mathrm{x}}$, considered as function from $\mathrm{x}(U)$ to the set of functions from 
 $\mathds{R}^3$ to $\mathds{R}$, belongs to 
 \begin{eqnarray*}
  C[\mathrm{x}(U), L^2_{z+\vert b \vert}(\mathds{R}^3)],
 \end{eqnarray*}
where the space $L^2_z(\mathds{R}^3)$ is the weighted $L^2$-space corresponding to the norm
\begin{eqnarray*}
 \Vert h \Vert_{L^2_z}=\left(\int_{\mathds{R}^3}(1+\bar{p}^2)^z\vert h(\bar{p}) \vert^2 d\bar{p}\right)^{\frac12} .
\end{eqnarray*}
\end{deff}
We need also function spaces for the initial datum for the distribution function.
Let $(\bar{\mathrm{x}},U)$ be local coordinates
on a manifold $\Sigma$ (which will play the role of the initial hypersurface). Then we define local coordinates on $T\Sigma$ via 
$\bar{\Xi}_{\bar{\mathrm{x}}}(\bar{p}^{i}\partial_{\bar{\mathrm{x}}}^{i}\vert_{\bar{\zeta}})=(\bar{\mathrm{x}}(\bar{\zeta}),\bar{p})$. Moreover if $f$ is defined on 
$T\Sigma$, then we use the notation $\mathrm{\bar{f}}_{\bar{\mathrm{x}}}= {\bar{f}} \circ \bar{\Xi}_{\bar{\mathrm{x}}}^{-1}$. This permits us to introduce the analogous
function space for $\bar{f}$, cf. Definition 7.5 of \cite{Hans}:

\begin{deff}
 Let $\frac52 < z\in{\mathds{R}}$ and let $\Sigma$ be a hypersurface. The space $\bar{\mathcal{F}}^{\infty}_{z}(T\Sigma)$ is defined
 to consist of the smooth functions  $f: T\Sigma \rightarrow \mathds{R}$, such that for every choice of local coordinates $(\mathrm{x},U)$, $3$-multiindex $a$,
 $3$-multiindex $b$, the derivative $\partial_{\bar{\mathrm{x}}}^{a}\partial_{\bar{p}}^{b} \mathrm{\bar{f}}_{\bar{\mathrm{x}}}$, considered as function 
 from $\bar{\mathrm{x}}(U)$ to the set of functions from $\mathds{R}^3$ to $\mathds{R}$, belongs to 
 \begin{eqnarray*}
  C[\bar{\mathrm{x}}(U), L^2_{z+\vert b \vert}(\mathds{R}^3)].
 \end{eqnarray*}
\end{deff}
Let us now define what is meant by \emph{initial data}:
\begin{deff}\label{id}
 Let $\frac52 < z\in{\mathds{R}}$. Initial data for (\ref{Einstein})-(\ref{Vlasov}) consist of an oriented hypersurface $\Sigma$, 
 a function $\bar{f}\in \bar{\mathcal{F}}^{\infty}_{z}(T\Sigma)$, a Riemannian metric $\bar{g}$ and a symmetric 2-tensor field $\bar{k}$ on $\Sigma$ all assumed to be smooth
 and to satisfy
 \begin{eqnarray*}
  \bar{R}-\bar{k}_{ij}\bar{k}^{ij}+(\tr \bar{k})^2&=&2\Lambda+2\rho,\\
  \bar{\nabla}^j\bar{k}_{ij}-\bar{\nabla}_i(\tr \bar{k})&=&-\bar{J}_i,
 \end{eqnarray*}
where $\bar{\nabla}$ is the Levi-Civita connection of $\bar{g}$, $\bar{R}$ the associated scalar curvature, indices are raised and lowered by $\bar{g}$, and the energy denity $\rho$ 
and current $\bar{J}$ are given by:
\begin{eqnarray*}
 \rho(\bar{\zeta})&=&\int_{T_{\bar{\zeta}}\Sigma} \bar{f}(\bar{p}) \sqrt{1+\bar{g}(\bar{p},\bar{p})}\varpi_{\bar{\zeta},\bar{g}},\\
 \bar{J}(\bar{r})&=&\int_{T_{\bar{\zeta}}\Sigma} \bar{f}(\bar{p}) \bar{g}(\bar{p},\bar{r}) \varpi_{\bar{\zeta},\bar{g}},
\end{eqnarray*}
where $\bar{\zeta}\in{\Sigma}$, $\bar{r}\in{T_{\bar{\zeta}}\Sigma}$, $\varpi_{\bar{\zeta},\bar{g}}$  is the volume form on $T_{\bar{\zeta}}\Sigma$ induced by $\bar{g}$ and $\bar{p}\in{T_{\bar{\zeta}}\Sigma}$.
\end{deff}
The associated \emph{initial value problem} to these initial data is defined as follows (Definition 7.11 of \cite{Hans}):
\begin{deff}
 Given initial data $(\Sigma,\bar{g},\bar{k})$ as in Definition \ref{id}, the initial value problem is that of finding the triple $(M,g,f)$ to (\ref{Einstein})-(\ref{Vlasov}), and an embedding
 $i:\Sigma \rightarrow M$ such that
 \begin{eqnarray*}
  i^{*}g=\bar{g}, \   \ \bar{f}=i^*(f \circ \pr^{-1}_{i(\Sigma)}),
 \end{eqnarray*}
and if $k$ is the second fundamental form of $i(\Sigma)$, then $i^*k=\bar{k}$. Such triple is refered to as a development of the initial data, the existence of $i$ beeing tacit. If 
$i(\Sigma)$ is a Cauchy hypersurface in $(M,g)$, the triple is called a globally hyperbolic development.
\end{deff}
Moreover there exists a unique maximal globally hyperbolic development of corresponding initial data (cf. \cite{Hans}, Definition 7.14). 

Finally let us introduce a norm which is needed for the stability result (Definition 7.7 of \cite{Hans}):
\begin{eqnarray}\label{norm}
 \Vert \bar{f} \Vert_{H^{l}_{z}}(U)=\left(\sum_{\vert a \vert + \vert b \vert \leq l} \int_{\bar{x}(U)\times {\R^3}} (1+\vert \bar{p} \vert^2)^{z+\vert b \vert}\vert \partial^{a}_{\bar{\zeta}} \partial^{b}_{\bar{p}} \bar{f}_{\bar{x}}\vert^2 (\bar{\zeta},\bar{p})d\bar{\zeta}d\bar{p}\right)^{\frac12}. 
\end{eqnarray}


\section{The Einstein-Vlasov system with surface symmetry}

Surface symmetric spacetimes are defined on manifolds of the form $M=\mathbb{R} \times \mathbb{S} \times S_K$ with $S_K$ being a compact orientable surface. 
The surfaces diffeomorphic to $S_K$ are called surfaces of symmetry. Moreover the Lorentz metric $g$ is a globally hyperbolic metric on $M$ for which each submanifold 
$\{t\}\times\mathbb{S} \times S_K$ is a Cauchy hypersurface. The orbits of the symmetry action are two-dimensional spheres, flat tori and hyperbolic spaces
for the case $K = 1$, $K=0$ and $K=-1$ respectively. For details and a definition of surface symmetric spacetimes we refer to \cite{rendallcrush}.

\subsection{Basic equations}
 
The metric takes the following form cf. for instance \cite{Tchap1}:
\begin{equation}\label{metric}
  g = -e^{2\mu(t,r)}dt^2 + e^{2\lambda(t,r)}dr^2 + t^2 g_K,  
\end{equation}
with 
\begin{equation*}
 g_K= d\theta^2 + \sin_{K}^{2}\theta d\varphi^{2}.
\end{equation*}

Here $t > 0$, the functions $\lambda$ and $\mu$ are periodic in
$r$ with period $1$ and
\begin{displaymath}
 \sin_{K}\theta := \left\{ \begin{array}{ll}
\sin\theta & \textrm{if $K=1,$}\\
1 & \textrm{if $K=0,$}\\
\sinh\theta & \textrm{if $K=-1.$}
  \end{array} \right.
\end{displaymath}
It has been shown \cite{ARR,Rein} that due to the symmetry the distribution function 
can be written in terms of the variables $t$, $r \in [0,1]$, $w \in \R$ and $F\in [0,\infty)$ where
\begin{eqnarray*}
&&w = e^{\lambda} p^1, \\
&&F = t^{4}(p^2)^2 + t^4 \sin_K^2\theta (p^{3})^2.
\end{eqnarray*}
In this expression $p^{a}$ are the spatial components of $p$ and the quantity $F$ is conserved along geodesics. The coordinates $(\theta,\varphi)$ range in $[0, \pi] \times [0, 2\pi]$, $[0, 2\pi] \times [0, 2\pi]$ and $[0, \infty)\times [0, 2\pi]$ for $K=1$, $K = 0$ and $K=-1$ respectively. For a proof of this in the hyperbolic case
(which easily can be extended to the other cases) see the proof of Lemma 3.1 in \cite{ARR}.
In these variables we have $p^0 = e^{-\mu}\sqrt{1 + w^{2} + F/t^{2}}$. Since the term $\sqrt{1 + w^{2} + F/t^{2}}$ will appear very often let us denote it by $V$.
Using the variables $w$ and $F$ in the expression (\ref{em}) we have:
 \begin{eqnarray*}
T_{\alpha\beta}=\frac{\pi}{t^{2}}  \int_0^{\infty} \int_{-\infty}^{\infty} V^{-1} f(t,r,w,F) p_{\alpha}p_{\beta} dwdF.
\end{eqnarray*}
The complete Einstein-Vlasov system with a cosmological constant with these variables reads as follows, where a dot and a prime denote derivation 
of the metric components with respect to $t$ and $r$ respectively:
\begin{equation} \label{1}
\partial_{t}f + \frac{e^{\mu-\lambda}w}{V}
\partial_{r}f - (\dot{\lambda} w +
e^{\mu-\lambda}\mu'V)\partial_{w}f = 0,
\end{equation}
\begin{equation} \label{2}
e^{-2\mu} (2t\dot{\lambda}+1)+ K - \Lambda t^{2} = 8 \pi t^{2}\rho,
\end{equation}
\begin{equation} \label{3}
e^{-2\mu} (2t\dot{\mu}-1)- K + \Lambda t^{2} = 8 \pi t^{2}p,
\end{equation}
 
\begin{equation} \label{4}
\mu' = -4 \pi t e^{\lambda+\mu}j,
\end{equation}
\begin{equation}\label{5}
e^{-2\lambda}\left(\mu'' + \mu'(\mu' - \lambda')\right) -
e^{-2\mu}\left(\ddot{\lambda}+(\dot{\lambda}-
\dot{\mu})(\dot{\lambda}+\frac{1}{t})\right) + \Lambda  = 4 \pi q,
\end{equation}
where
\begin{eqnarray*} 
&&\rho(t, r) := \frac{\pi}{t^{2}} \int_{-\infty}^{\infty}\int_{0}^{\infty} V f(t, r, w, F) dF dw = e^{-2\mu}T_{00}(t, r),\\
&&p(t, r) := \frac{\pi}{t^{2}} \int_{-\infty}^{\infty}\int_{0}^{\infty} \frac{w^{2}}{V} f(t, r, w, F) dF dw = e^{-2\lambda}T_{11}(t, r),\\
&&j(t, r) := \frac{\pi}{t^{2}} \int_{-\infty}^{\infty}\int_{0}^{\infty} w f(t, r, w, F) dF dw = -e^{\lambda +\mu}T_{01}(t, r),\\ 
&&q(t, r) := \frac{\pi}{t^{4}} \int_{-\infty}^{\infty}\int_{0}^{\infty} \frac{F}{V} f(t, r, w, F)dF dw = \frac{2}{t^{2}}T_{22}(t, r).
\end{eqnarray*}

\subsection{Asymptotic behaviour of the first derivatives of the metric terms}

In Lemma 6.1 of \cite{Tchap2} it was shown that for any characteristic
the quantity $u=t w$ converges uniformly to a constant along the characteristics, which implies together with
the fact that $f(t_0,r,w,F)$ is compactly supported in $w$ that there exist a constant $C$ such that
\begin{equation}\label{w}
  |w| \leq Ct^{-1} \ \textrm{and} \ f(t,r,w,F) =0, \ \textrm{if} \ |w| \geq  Ct^{-1}.
\end{equation}
The following estimates were determined for the metric for $K=0$ and $K=-1$, cf. (3.15)-(3.19) in \cite{Tchap1}, and for $K = 1$, cf. (5.7)-(5.11) in \cite{Tchap2}:
\begin{eqnarray*}
&&\dot\lambda=+t^{-1}(1+O(t^{-2})), \ \lambda=+ \ln t [1+O((\ln t)^{-1})],\\
&&\dot{\mu}=-t^{-1}(1+O(t^{-2})), \ \mu=-\ln t [1+O((\ln t)^{-1})],\\
&&\mu'= O(t^{-3+\varepsilon}),
\end{eqnarray*}
with $\varepsilon\in{(0,\frac23)}$. 
For $e^{\mu}$ one has actually a better estimate, namely the estimate (3.14) of \cite{Tchap1}: $e^{\mu}=\sqrt{\frac{3}{\Lambda}} t^{-1}(1+O(t^{-2}))$.
For the matter terms, cf. Proposition 6.2 in \cite{Tchap2} it was found that for the three types of surface symmetry we have that
\begin{eqnarray*}
 \rho=O(t^{-3}), \  p=O(t^{-5}),\  j=O(t^{-4}),\  q=O(t^{-5}).
\end{eqnarray*}
From the estimate for $j$ we thus obtain 
\begin{eqnarray*}
 \mu'= O(t^{-3}).
\end{eqnarray*}

\subsection{Characteristics}

\begin{lem}\label{c}
Consider a solution to the Einstein-Vlasov system with a positive cosmological constant $\Lambda$ and surface symmetry and fix $t_0\in(0,\infty)$ for $K\leq 0$ 
 and fix $t_0\in(\Lambda^{-\frac12},\infty)$ for $K=1$. Then there is a positive constant $C$, depending on the solution, such that
  \begin{eqnarray*} 
 &&\dot{\xi}=s^{-1}(1+c_1)\xi+(1+c_2)\eta,\\
 &&\dot{\eta}=s^{-2}(1+c_3)\xi+s^{-1}c_4\eta,
\end{eqnarray*}
where 
\begin{eqnarray}\label{C} 
\sum\limits_{i=1}^4 \vert c_i \vert \leq Cs^{-2},
\end{eqnarray}
and the estimates hold for all $(t,r,w)$ in the support of $f$ and $s\in[t_0,t]$.
\end{lem}
\begin{proof}
 The characteristics $(R,W)(\cdot, t, r, w, F)$ of our system are the solution to the following system:
\begin{eqnarray*} 
  &&\frac{dR}{ds}=\frac{e^{\mu-\lambda}W}{V},\\
 &&\frac{dW}{ds}=-\dot{\lambda}W-\mu'e^{\mu-\lambda}V,
\end{eqnarray*}
with initial data
\begin{eqnarray*}
 R(t, t, r, w, F)=r, \ W(t, t, r, w, F)=w.
\end{eqnarray*}

Differentiating the system leads to derivatives we do not control, however a certain combination was found in \cite{Rein} (Lemma 3.2), where this is the case.
Consider the following quantities of \cite{Rein}: 
\begin{eqnarray}\label{xi} 
 &&\xi= e^{\lambda-\mu} \partial R,\\
 \label{eta}&&\eta=\partial W+V e^{\lambda-\mu}\dot{\lambda}\partial R,
\end{eqnarray}
for $\partial\in\{\partial_t,\partial_r,\partial_w\}$. One can compute the derivatives of these quantities obtaining:
\begin{eqnarray*} 
 &&\frac{d\xi}{ds}=[\dot{\lambda}-\dot{\mu}+(\lambda'-\mu')\frac{dR}{ds}]e^{\lambda-\mu}\partial R+e^{\lambda-\mu}\partial \frac{dR}{ds},\\
 &&\frac{d\eta}{ds}=\partial \frac{dW}{ds}+V e^{\lambda-\mu}\{[V^{-2}(W\frac{dW}{ds}-\frac{F}{s^3})+\dot{\lambda}-\dot{\mu}+(\lambda'-\mu')\frac{dR}{ds}]\dot{\lambda}\partial R+(\ddot{\lambda}+\dot{\lambda}'\frac{dR}{ds})\partial R+\dot{\lambda}\partial \frac{dR}{ds}\}.
\end{eqnarray*}
The derivatives $\partial \frac{d}{ds}(R,W)$ are the following:
\begin{eqnarray*} 
 &&\partial \frac{dR}{ds}=e^{\mu-\lambda}V^{-3}[ V^2(\mu'-\lambda')W\partial R + (1+F/s^2)\partial W],\\
 &&\partial \frac{dW}{ds}=-[\dot{\lambda}'W+(\mu''+\mu'(\mu'-\lambda'))e^{\mu-\lambda}V]\partial R-[\dot{\lambda}+W V^{-1}e^{\mu-\lambda}\mu']\partial W.
\end{eqnarray*}
Thus
\begin{eqnarray*} 
&& \frac{d\xi}{ds}=(\dot{\lambda}-\dot{\mu})e^{\lambda-\mu}\partial R+\frac{1+F/s^2}{V^3} \partial W,\\
&& \frac{d\eta}{ds}= \{ V [e^{\lambda-\mu}(\ddot{\lambda}+\dot{\lambda}(\dot{\lambda}-\dot{\mu}))-e^{\mu-\lambda}(\mu''+\mu'(\mu'-\lambda'))]-\dot{\lambda}\mu'W-V^{-1}e^{\lambda-\mu}\dot{\lambda}(\dot{\lambda}W^2+F/s^3)\}\partial R\\
&&-WV^{-1}(e^{\mu-\lambda}\mu'+W V^{-1}\dot{\lambda})\partial W,
\end{eqnarray*}
and in terms of $\xi$ and $\eta$
\begin{eqnarray*} 
 &&\frac{d\xi}{ds}=\left(\frac{W^2}{V^2}\dot{\lambda}-\dot{\mu}\right)\xi+\frac{1+F/s^2}{V^3}\eta,\\
 &&\frac{d\eta}{ds}=\{ V[(\ddot{\lambda}+\dot{\lambda}(\dot{\lambda}-\dot{\mu}))-e^{2\mu-2\lambda}(\mu''+\mu'(\mu'-\lambda'))]\\
 &&-V^{-1}\dot{\lambda}F/s^3\}\xi-WV^{-1}(e^{\mu-\lambda}\mu'+WV^{-1}\dot{\lambda})\eta.
\end{eqnarray*}
Now using equation (\ref{5}) in the second equation for the second derivatives we obtain:
\begin{eqnarray}\label{mam} 
 &&\frac{d\xi}{ds}=(V^{-2} W^2\dot{\lambda}-\dot{\mu})\xi+\frac{1+F/s^2}{V^3}\eta,\\
 \label{mam2}&&\frac{d\eta}{ds}=\{ V[(\Lambda-4\pi q)e^{2\mu}+(\dot{\mu}-\dot{\lambda})s^{-1}]-V^{-1}\dot{\lambda}F/s^3\}\xi-WV^{-1}(e^{\mu-\lambda}\mu'+WV^{-1}\dot{\lambda})\eta.
\end{eqnarray}
Using the estimates the lemma follows.
\end{proof}

\begin{lem}\label{ch}
 Consider a solution to the Einstein-Vlasov system with a positive cosmological constant $\Lambda$ and surface symmetry and fix $t_0\in(0,\infty)$ for $K\leq 0$ 
 and fix $t_0\in(\Lambda^{-\frac12},\infty)$ for $K=1$. Then there is a positive constant $C$,
 depending on the solution, such that
 \begin{eqnarray*} 
  \vert \frac{\partial R}{\partial r}(s,t,r,w,F) \vert + s\vert \frac{\partial W}{\partial r}(s,t,r,w,F)\vert \leq C,
   \end{eqnarray*}
 and the estimates hold for all $(t,r,w)$ in the support of $f$ and $s\in[t_0,t]$.
\end{lem}
\begin{proof}
               Let $\hat{\eta}=s \eta$. Then due to Lemma \ref{c} we have:
 \begin{eqnarray*} 
 &&\dot{\xi}=s^{-1}[(1+c_1)\xi+(1+c_2)\hat{\eta}],\\
 &&\dot{\hat{\eta}}=s^{-1}[(1+c_3)\xi+(1+c_4)\hat{\eta}].
\end{eqnarray*}
Consider the quantity
\begin{eqnarray*} 
 E=s^{-4}(\xi+\hat{\eta})^2+(\xi-\hat{\eta})^2.
\end{eqnarray*}
Then
\begin{eqnarray*} 
 \frac{dE}{ds}=2s^{-5}[(c_1+c_3)\xi+(c_2+c_4)\hat{\eta}](\xi+\hat{\eta})+2s^{-1}[(c_1-c_3)\xi-(c_4-c_2)\hat{\eta}](\xi-\hat{\eta}).
\end{eqnarray*}
We obtain:
\begin{eqnarray}\label{ie1} 
 \frac{dE}{ds}\geq- C_1 s^{-1} E,
\end{eqnarray}
where  
\begin{eqnarray*} 
 C_1= 2\sup(\vert c_1+c_3 \vert, \vert c_2+c_4 \vert,  \vert c_1-c_3\vert, \vert c_4-c_2 \vert), 
\end{eqnarray*}
and 
\begin{eqnarray*} 
 \vert C_1 \vert \leq C s^{-2},  
\end{eqnarray*}
due to (\ref{C}) of Lemma \ref{c}. As a consequence of (\ref{ie1}) the following inequality holds:

\begin{eqnarray}
 \frac{dE}{ds}\geq- C s^{-3} E.
\end{eqnarray}
Integrating this inequality leads to:
\begin{eqnarray*} 
 E(s;t,r,w)\leq C E(t;t,r,w),
\end{eqnarray*}
for $s\in{[t_0,t]}$. We are interested in the derivatives with respect to $r$. Thus we will consider (\ref{xi})-(\ref{eta}) only with $\partial=\partial_r$ in the following. 
The quantity $E(t;t,r,w)$ can be estimated having in mind that
\begin{eqnarray*}
 &&\frac{\partial R}{\partial r}(t;t,r,w,F)=1,\\
 &&\frac{\partial W}{\partial r}(t;t,r,w,F)=0,
\end{eqnarray*}
and using the estimates for the metric terms and their derivatives. What we obtain is
\begin{eqnarray*}
 E(t;t,r,w)\leq C,
\end{eqnarray*}
where we have used the fact that
\begin{eqnarray*}
 \xi(t;t,r,w)-\hat{\eta}(t;t,r,w)=e^{\lambda-\mu}(1-V\dot{\lambda}t)=O(1).
\end{eqnarray*}
We have thus that:
\begin{eqnarray}\label{a} 
 &&\left\vert \frac{\partial R}{\partial r}(1+V\dot{\lambda}s)e^{\lambda-\mu}+s\frac{\partial W}{\partial r}\right\vert \leq Cs^2, \\
 \label{b}&&\left\vert \frac{\partial R}{\partial r}(1-V\dot{\lambda}s)e^{\lambda-\mu}-s\frac{\partial W}{\partial r}\right\vert \leq C,
\end{eqnarray}
which has the consequence that
\begin{eqnarray*} 
 \left\vert \frac{\partial R}{\partial r}\right\vert \leq Cs^2e^{\mu-\lambda} \leq C,
\end{eqnarray*}
where the last inequality was obtained using the estimates. Using this inequality in (\ref{b}) implies that:
\begin{eqnarray*} 
 \left\vert s\frac{\partial W}{\partial r}\right\vert\leq C,
\end{eqnarray*}
 and the lemma is proved.
              \end{proof}

\section{Higher order estimates}
\subsection{First step}

\begin{lem}\label{k}
 Consider a solution to the Einstein-Vlasov system with a positive cosmological constant $\Lambda$ and surface symmetry and fix $t_0\in(0,\infty)$ for $K\leq 0$ 
 and fix $t_0\in(\Lambda^{-\frac12},\infty)$ for $K=1$. Then there is a positive constant $C$, depending on the solution and $t_0$, such that
 \begin{eqnarray*} 
 \Vert \rho'\Vert_{C^0}+t\Vert j'\Vert_{C^0}+t^2\Vert q'\Vert_{C^0}+t^2\Vert p'\Vert_{C^0} \leq C t^{-3}, 
\end{eqnarray*}
for all $t \geq t_0$.
\end{lem}
\begin{proof}
 It is a direct consequence of Lemma \ref{ch} and the fact that $w=O(t^{-1})$ in the support of $f$.
\end{proof}

Now we are able to obtain higher estimates.
\begin{lem}\label{l}
 Consider a solution to the Einstein-Vlasov system with a positive cosmological constant $\Lambda$ and surface symmetry and fix $t_0\in(0,\infty)$ for $K\leq 0$ 
 and fix $t_0\in(\Lambda^{-\frac12},\infty)$ for $K=1$. Then there is a positive constant $C$, depending on the solution and $t_0$, such that
 \begin{eqnarray*} 
t\Vert \dot{\lambda}'\Vert_{C^0}+t \Vert \dot{\mu}'\Vert_{C^0}+ \Vert \mu''\Vert_{C^0} &\leq& Ct^{-3},\\
\Vert \lambda' \Vert_{C^0} &\leq& C,
\end{eqnarray*}
for all $t \geq t_0$.
\end{lem}
\begin{proof}
 From (\ref{2})-(\ref{3}) we have 
\begin{eqnarray}\label{misch}
&&          \dot{\lambda}'=[4 \pi t^{2}\rho'+\mu'(8 \pi t^{2}\rho-K+\Lambda t^{2})]e^{2\mu}t^{-1},\\
&&\label{mu}\dot{\mu}' =[4 \pi t^{2}p'+\mu'( 8 \pi t^{2}p+K-\Lambda t^{2})]e^{2\mu}t^{-1},
\end{eqnarray}
which means that
\begin{eqnarray*}
\dot{\lambda}'=O(t^{-4}),
\end{eqnarray*}
and integrating in time:
\begin{eqnarray*}
 \lambda'=O(1).
\end{eqnarray*}
From (\ref{4}) we obtain:
\begin{eqnarray*} 
&& \mu''=\mu'(\mu'+\lambda')-4 \pi t e^{\lambda+\mu}j',
\end{eqnarray*}
and the rest of the lemma follows.
\end{proof}

\subsection{Second step: higher order derivatives of the characteristic system}

\begin{lem}\label{fff}
 Consider a solution to the Einstein-Vlasov system with a positive cosmological constant $\Lambda$ and surface symmetry and fix $t_0\in(0,\infty)$ for $K\leq 0$ 
 and fix $t_0\in(\Lambda^{-\frac12},\infty)$ for $K=1$,  then for all $n\in{\{0\}\cup \mathbb{N}}$ there is a positive constant $C_n$, depending on the solution and $t_0$, such that 
 \begin{eqnarray*}
  \Vert \lambda'\Vert_{C^{n}}+t^4\Vert \dot{\lambda}' \Vert_{C^{n}}+t^3\Vert \mu''\Vert_{C^{n}}+t^4\Vert \dot{\mu}'\Vert_{C^{n}} &\leq& C_n,\\
\label{mot} \Vert \rho\Vert_{C^{n}}+t\Vert j\Vert_{C^{n}}+t^2\Vert q\Vert_{C^{n}}+t^2\Vert p\Vert_{C^{n}} &\leq& C_n t^{-3},
\end{eqnarray*}
for $t \geq t_0$.
\end{lem}
\begin{proof}
 We prove this by induction.\\
\textbf{Inductive assumption.} For some $ N \in{\mathbb{N}}$ there are constants $C_j$, $j=1,...,N$ depeding on the solution and $t_0$ such that
\begin{eqnarray}\label{i1} 
 \vert \frac{\partial^j R}{\partial r^j}(s,t,r,w,F) \vert + s\vert \frac{\partial^j W}{\partial r^j}(s,t,r,w,F)\vert &\leq& C_j,\\
 \label{met} \Vert \lambda'\Vert_{C^{N-2}}+t^4\Vert \dot{\lambda} \Vert_{C^{N-1}}+t^3\Vert \mu'\Vert_{C^{N-1}}+t^4\Vert \dot{\mu}\Vert_{C^{N-1}} &\leq& C_{N}.  
\end{eqnarray}
The assumption holds for $N=2$ due to Lemmas \ref{ch} and \ref{l}.
From (\ref{i1}) we have $\vert \frac{{\partial}^j f}{\partial r^j} \vert \leq C_j$ for $j=1,\dots,N$. Thus (\ref{mot}) holds for $n=N$. 
Let us take $N-1$ derivatives of (\ref{misch}). The $N-1$-th derivative of $\dot{\lambda}$ is bounded due to the induction assumption (\ref{met}). The same bound will hold
for the $N$-th derivative of $\dot{\lambda}$ since the $N$-th derivative of $\mu$ is bounded by (\ref{met}) and the $N$-th derivative of $\rho$ by (\ref{i1}). Integrating in time the 
bound for the $N$-th derivative of $\dot{\lambda}$ we obtain that the $N$-th derivative of $\lambda$ is bounded. Taking $N-1$ derivatives of (\ref{mu}) and due to the bounds
on the $N$th derivative of $\mu$ and $p$ we obtain the bound for the $N$-th derivative of $\dot{\mu}$. Now we can proceed to take $N$ 
derivatives of (\ref{4}) with the result that the $N+1$th derivative of $\mu$ is bounded. It remains to show that (\ref{i1}) is valid for $j=N+1$.
Denote for the rest of this section the number of partial derivatives with respect to $r$ with a subindex in the case of the quantities $\xi$, $\eta$, $R$ and $W$, e.g.:
\begin{eqnarray*}
 \xi_N\equiv\frac{\partial^N \xi}{\partial r^N} \ \ etc.
\end{eqnarray*}
Taking $N$-derivatives of (\ref{mam})-(\ref{mam2}) we obtain that the right hand side consists of the same terms with $\xi$, $\eta$ replaced by $\xi_N$, $\eta_N$ plus terms, let us
call them $r_1$ and $r_2$ respectively, which arise from taking at least one, up to $N$ derivatives of the coefficients in front of $\xi$, $\eta$ multiplied by $\xi_j$, $\eta_j$ respectively with $j$ ranging from $0$
to $N-1$. Taking $N-1$-derivatives of the definitions of $\xi$ and $\eta$ (\ref{xi})-(\ref{eta}) and having in mind that due to (\ref{i1}) we have bounds on
$R_j$ and $W_j$ for $j=1,...,N$, we obtain:
\begin{eqnarray*}
&&\xi_{j-1}=O(s^2),\\
&&\eta_{j-1}=O(s).
\end{eqnarray*}
Thus we can estimate the terms $r_1$ and $r_2$ obtaining:
\begin{eqnarray*}
 &&r_1=O(s^{-1}),\\
 &&r_2=O(s^{-2}).
\end{eqnarray*}
Having in mind that we have shown that (\ref{met}) holds for $N+1$:
\begin{eqnarray*} 
 &&\dot{\xi}_N=s^{-1}(1+c_1)\xi_N+(1+c_2)\eta_N+O(s^{-1}),\\
 &&\dot{\eta}_N=s^{-2}(1+c_3)\xi_N+s^{-1}c_4\eta_N+O(s^{-2}),
\end{eqnarray*}
Introducing the notation ${\hat{\eta}}_N=s\eta_N$ we obtain:
\begin{eqnarray*} 
 &&\dot{\xi}_N=s^{-1}[(1+c_1)\xi_N+(1+c_2)\hat{\eta}_N]+O(s^{-1}),\\
 &&\dot{\hat{\eta}}_N=s^{-1}[(1+c_3)\xi+(1+c_4)\hat{\eta}_N]+O(s^{-1}).
\end{eqnarray*}
Define
\begin{eqnarray*} 
 E_N=s^{-4}(\xi_N+\hat{\eta}_N)^2+(\xi_N-\hat{\eta}_N)^2.
\end{eqnarray*}
Then
\begin{eqnarray*}
 \frac{dE_N}{ds}\geq -C_N s^{-3}E_N-C s^{-1}E_N^{\frac12},
\end{eqnarray*}
where $C_N$ is a positive constant. Define $\bar{E}_N=e^{-f_N(s)}E_N$ with $f_{N}(t_0)=0$ and $f_{N}'=-C_N s^{-3}$.
Then,
\begin{eqnarray*}
  \frac{d\bar{E}_N}{ds}\geq -C s^{-1}\bar{E}_N^{\frac12}.
\end{eqnarray*}
Integrating we have
\begin{eqnarray*}
 \bar{E}_N^{\frac12}(s;t,r,w)\leq\bar{E}_N^{\frac12}(t;t,r,w)+ C,
\end{eqnarray*}
for $s\in{[t_0,t]}$. We thus obtain that
\begin{eqnarray*}
 E_N(s;t,r,w) \leq C.
\end{eqnarray*}
Taking $N$ derivatives of the definitions of $\xi$ and $\eta$ (\ref{xi})-(\ref{eta}) we come to the conclusion that
\begin{eqnarray}\label{ex}
 &&\xi_N= e^{\lambda-\mu} R_{N+1}+O(s^2),\\
 \label{ex2}&&\hat{\eta}_N=sW_{N+1}+sV e^{\lambda-\mu}\dot{\lambda}R_{N+1}+O(s^2),
\end{eqnarray}
and 
\begin{eqnarray*}
 R_{N+1}(s;t,r,w)\leq C.
\end{eqnarray*}
The expressions (\ref{ex})-(\ref{ex2}) are not enough to obtain the desired estimate for $W_{N+1}$, but note that
\begin{eqnarray*}
  \xi_N-\hat{\eta}_N = O(1)+s W_{N+1}+ (1-V\dot{\lambda}s)\sum_{j=1}^{N}R_{N+1-j}\frac{\partial^j}{\partial r^j}(e^{\lambda-\mu}), 
\end{eqnarray*}
where the last term on the right hand side is also of order $O(1)$ which leads us to the desired estimate for $W_{N+1}$ and (\ref{i1}) holds for $j=N+1$.
\end{proof}

\section{Geometric consequences and cosmic no-hair theorem}
 
 \begin{thm}\label{geo}
  Consider a surface symmetric solution to the Einstein-Vlasov system with a positive cosmological constant $\Lambda$. Choose coordinates such that the corresponding metric
  takes the form (\ref{metric}) on $I\times \mathbb{S} \times S_K$ where $I=(t_0,\infty)$ with $t_0\geq 0$ for 
 $K\leq 0$ and $t_0\geq \Lambda^{-\frac12}$ for $K=1$. Let $t_1=t_0+2$. Denote by $\bar{g}(t,\cdot)$ and $\bar{k}(t,\cdot)$ the metric and second fundamental form induced
 by $g$ on the hypersurface $\{t\}\times \mathbb{S} \times S_K$, and 
$\bar{g}_{ij}(t,\cdot)$ and $\bar{k}_{ij}(t,\cdot)$ denote the components of $\bar{g}(t,\cdot)$ and $\bar{k}(t,\cdot)$ with respect to the vectorfields 
$\partial_1=\partial_r$, $\partial_2=\partial_{\theta}$ and $\partial_3=\partial_{\varphi}$ Then there is smooth function $\lambda_{\infty}$ on $\mathbb{S}$ 
and for all $n\in{\{0\}\cup \mathbb{N}}$ there is a constant $C_n$ such that
 \begin{eqnarray}\label{lambda}
  \left\Vert e^{2\lambda(t,\cdot)-2\ln t}-e^{2\lambda_{\infty}} \right\Vert_{C^n} \leq C_n t^{-3},\\
\label{guess} t \left\Vert t^{-2}\bar{g}_{ij}(t,\cdot)-\bar{g}_{\infty,ij} \right\Vert_{C^n} +\left\Vert t^{-2} \bar{k}_{ij}(t,\cdot)-H\bar{g}_{\infty,ij} \right\Vert_{C^n} \leq C_n t^{-2},
 \end{eqnarray}
 where  
 \begin{eqnarray*}
  \bar{g}_{\infty}=e^{2\lambda_{\infty}}dr^2 +g_K,
 \end{eqnarray*}
 and $H=\sqrt{\frac\Lambda3}$.
 \end{thm}
\begin{proof}

 The inequality (\ref{lambda}) and the inequality concerning $\bar{g}_{ij}$ is a direct consequence of Lemma \ref{fff}. 
  The induced metric is 
 \begin{eqnarray*}
  \bar{g}=e^{2\lambda}dr^2+t^2g_K,
 \end{eqnarray*}
 and
 \begin{eqnarray*}
  \bar{k}_{ij}=\frac12 e^{-\mu}\partial_t \bar{g}_{ij}.
 \end{eqnarray*}
 Since
  \begin{eqnarray*}
   \partial_t \bar{g}=2\dot{\lambda}e^{2\lambda}dr^2+2tg_K,
  \end{eqnarray*}
and using the estimate of
 $e^{\mu}$
 \begin{eqnarray*}
 e^{\mu}=H^{-1}t^{-1}(1+O(t^{-2})),
\end{eqnarray*}
from Lemma \ref{fff} we also obtain that
 \begin{eqnarray*}
  \Vert \bar{k}_{ij} - H \bar{g}_{ij} \Vert_{C^n} \leq C_n,
 \end{eqnarray*}
 which completes the proof.
\end{proof}
In \cite{Tchap2,Tchap1} a form of cosmic no-hair theorem was already proven. We will follow the arguments of \cite{HakanHans} to obtain a stronger result.

Recall that $\Sigma_t=\{t\}\times \mathbb{S} \times S_K$ is a Cauchy hypersurface for each $(t_0,\infty)$. Let $\gamma=(\gamma^0,\bar{\gamma})$ be a future directed and inextendible causal curve with $\gamma^0(s)=s$ and $\bar{\gamma}(s)=[r(s),\theta(s),\varphi(s)]$
 defined on the interval $I_{\gamma}=(t_0,\infty)$. Due to the causality and having in mind Theorem \ref{geo} we have that there is a constant independent of the curve $K_0>1$
 such that:
 \begin{eqnarray*}
  \bar{g}_{\infty,ij}[\bar{\gamma}(t)]\dot{\bar{\gamma}}^i(t)\dot{\bar{\gamma}}^j(t) \leq K_0^2 H^{-2}t^{-4}, \ \ \ \forall t\geq t_1=t_0+2.
 \end{eqnarray*}
Thus there is a $\bar{x}_0\in{\mathbb{S}\times S_K}$ such that $d_{\infty}[\bar{\gamma}(t),\bar{x}_0]\leq K_0(Ht)^{-1}$ for all $t\geq t_1$, where $d_{\infty}$ is the
topological metric on $\mathbb{S}\times S_K$ induced by $\bar{g}_{\infty}$. Let $\epsilon_{\inf}>0$ denote the injectivity radius of $(\mathbb{S}\times S_K,\bar{g}_{\infty})$. For a given
$\bar{x}\in{\mathbb{S}\times S_K}$ there are thus geodesic normal coordinates on $B_{\epsilon_{\inf}}$ where distances are computed with $d_{\infty}$. Fix $t_- >1+K_0 (H\epsilon_{\inf})^{-1}$.
Then
\begin{eqnarray*}
 J^{-}(\gamma)\cap J^{+}(\Sigma_{t_{-}})\subseteq D_{t_{-},K_{0},\bar{x}_0}=\{(t,\bar{x})\in{I\times \mathbb{S} \times S_K}:t \geq t_-, \ d_{\infty}(\bar{x},\bar{x}_0)\leq K_0 (Ht)^{-1}\}.
\end{eqnarray*}
Moreover the closed ball of radius $K_0 (Ht)^{-1}$ with respect to $d_{\infty}$ and centre $\bar{x}_0$ is contained in the domain of definition of geodesic normal coordinates
$\mathrm{x}$ with center at $\bar{x}_0$.
Defining $\psi(\tau,\bar{\zeta})=[e^{H\tau},\mathrm{x}^{-1}(\bar{\zeta})]$ we have
\begin{eqnarray*}
 \psi^{-1}(D_{t_{-},K_{0},\bar{x}_0})=\{(\tau,\bar{\zeta})\in{(H^{-1}\ln t_0,\infty)\times \mathbb{R}^3}:\tau \geq H^{-1} \ln t_-, \ \ \vert\bar{\zeta}\vert\leq K_0 H^{-1}e^{-H\tau}\}.
\end{eqnarray*}
 Letting $T$ be slightly smaller than $H^{-1} \ln t_-$ and $K$ be slightly larger than $K_0$, the map $\psi$ is still defined on the set
 \begin{eqnarray*}
  C_{\Lambda,K,T}=\{(\tau,\bar{\zeta}):\tau > T, \ \ \vert\bar{\zeta}\vert<R\},
\end{eqnarray*}
with $R=KH^{-1}e^{-H\tau}$ and thus let us define $D=\psi(C_{\Lambda,K,T})$. Let $\bar{\mathrm{g}}_{\infty,ij}$, $\bar{\mathrm{g}}_{ij}(\tau,\cdot)$ 
and $\bar{\mathrm{k}}_{ij}(\tau,\cdot)$ denote the components of $\bar{g}_{\infty}$, $\bar{g}(e^{H\tau},\cdot)$ and $\bar{k}(e^{H\tau},\cdot)$, respectively,
with respect to the geodesic normal coordinates $\bar{\mathrm{x}}$. Moreover consider them to be functions on the image $\bar{\mathrm{x}}$, i.e. on $B_{\epsilon_{\inf}}(0)$
with the origin corresponding to $\bar{x}_0$. Denote $S_{\tau}=\{\tau\} \times B_{R(\tau)}(0)$ and the metric of the de Sitter spacetime by
 \begin{eqnarray*}
  g_{\mathrm{dS}}=-d{\tau}^2+e^{2H\tau}g_E,
 \end{eqnarray*}
where $g_E$ is the standard flat Euclidean metric. Then using the previous theorem 
\begin{eqnarray*}    
   e^{H\tau} \left\Vert e^{-2H\tau}\bar{\mathrm{g}}_{ij}(\tau,\cdot)-\bar{\mathrm{g}}_{\infty,ij} \right\Vert_{C^n} +\left\Vert e^{-2H\tau} \bar{\mathrm{k}}_{ij}(\tau,\cdot)-H\bar{\mathrm{g}}_{\infty,ij} \right\Vert_{C^n} \leq C_n  e^{-2H\tau},
 \end{eqnarray*}
  for all $\tau\geq T$. By definition of the geodesic normal coordinates $\bar{\mathrm{x}}$:
 \begin{eqnarray*}
  \bar{\mathrm{g}}_{\infty,ij}(0)=\delta_{ij}, \ \ \ (\partial_l \bar{\mathrm{g}}_{\infty,ij})(0)=0.
 \end{eqnarray*}
From this follows for $\zeta \in{S_{\tau}}$ and $\tau \geq T$ that
\begin{eqnarray*}
 \vert ( \partial_l \bar{\mathrm{g}}_{\infty,ij})(\bar{\zeta}) \vert = \left\vert \int^1_0 \frac{d}{ds} [(\partial_l \bar{\mathrm{g}}_{\infty,ij})(s\bar{\zeta})]ds\right\vert \leq C e^{-H\tau},
\end{eqnarray*}
and
\begin{eqnarray*}
 \vert (\bar{\mathrm{g}}_{\infty,ij})(\bar{\zeta})-\delta_{ij} \vert \leq e^{-3H\tau}.
\end{eqnarray*}
Thus
\begin{eqnarray*}    
   e^{H\tau} \left\Vert e^{-2H\tau}\bar{\mathrm{g}}_{ij}(\tau,\cdot)-\delta_{ij} \right\Vert_{C^0} +\left\Vert e^{-2H\tau} \bar{\mathrm{k}}_{ij}(\tau,\cdot)-H\delta_{ij} \right\Vert_{C^0} \leq C  e^{-2H\tau},
 \end{eqnarray*}
  for all $\tau\geq T$. We can summarize this in the following theorem:
 \begin{thm}
  Consider a surface symmetric solution to the Einstein-Vlasov system with a positive cosmological constant $\Lambda$. Choose coordinates such that the corresponding metric
  takes the form (\ref{metric}) on $I\times \mathbb{S} \times S_K$ where $I=(t_0,\infty)$ . Then
  \begin{itemize}
   \item there is an open set $D$ in $(M,g)$, such that $J^{-}(\gamma)\cap J^{+}(\Sigma_{t_{-}})\subseteq D$, and $D$ is diffeomorphic to $C_{\Lambda,K,T}$ for a suitable
   choice of $K\geq 1$ and $T>0$
   \item using $\psi: C_{\Lambda,K,T} \mapsto D$ to denote the diffeomorphism and using $\bar{g}_{{\mathrm{dS}}}(\tau,\cdot)$ and $\bar{k}_{{\mathrm{dS}}}(\tau,\cdot)$ to denote the metric and second fundamental
   form induced on $S_{\tau}$ by $g_{{\mathrm{dS}}}$; using $\bar{g}(\tau,\cdot)$ and $\bar{k}(\tau,\cdot)$ to denote the metric and second fundamental form induced on $S_{\tau}$ by $\psi^* g$; and letting
   $N\in{N}$, the following holds:
   \begin{eqnarray}\label{0}
    e^{H\tau}\Vert \bar{g}_{{\mathrm{dS}}}(\tau,\cdot) - \bar{g}(t,\cdot)\Vert_{C^N_{{\mathrm{dS}}}(S_{\tau})} + \Vert \bar{k}_{{\mathrm{dS}}}(\tau,\cdot) - \bar{k}(\tau,\cdot)\Vert_{C^N_{{\mathrm{dS}}}(S_{\tau})}\leq e^{-2H\tau},
   \end{eqnarray}
for all $\tau>T$ where we use the notation
\begin{eqnarray*}
 \Vert h \Vert_{C^N_{{\mathrm{dS}}}(S_{\tau})}=\left( \sup_{S_{\tau}} \sum^N_{l=0} \bar{g}_{{\mathrm{dS}},i_1j_1}\dots\bar{g}_{{\mathrm{dS}},i_lj_l}\bar{g}_{{\mathrm{dS}}}^{im}\bar{g}_{{\mathrm{dS}}}^{jn}\bar{\nabla}^{i_1}_{{\mathrm{dS}}}\dots \bar{\nabla}_{{\mathrm{dS}}}^{i_l}h_{ij}\bar{\nabla}^{j_1}_{{\mathrm{dS}}}\dots\bar{\nabla}_{{\mathrm{dS}}}^{j_l}h_{mn}\right)^{\frac12},
\end{eqnarray*}
for a covariant 2-tensor field $h$ on $S_{\tau}$, where $\bar{\nabla}_{{\mathrm{dS}}}$ denotes the Levi-Civita connection associated with $\bar{g}_{{\mathrm{dS}}}(\tau,\cdot)$.
  \end{itemize}

 \end{thm}
The conclusions about the space-time of the theorem coincide with Definition 8 of \cite{HakanHans} which defines when a space-time is said to be future asymptotically de Sitter like, except that there it is sufficient that the left hand side of (\ref{0})
without the prefactor in the first term tends to zero in the limit $\tau \rightarrow \infty$.

\section{Energy estimates for the distribution function}
\subsection{Energy of order zero}
Let us use the following notation:
\begin{eqnarray}\label{alp}
&&\alpha=\frac{e^{\mu-\lambda}w}{V},\\
\label{be}&&\beta=\dot{\lambda} w +e^{\mu-\lambda}\mu'V.
\end{eqnarray}
Consider now a generalized Vlasov equation:
\begin{eqnarray}\label{General}
 \partial_t h  +\alpha \partial_{r}h - \beta \partial_{w}h = R.
\end{eqnarray}
Define
\begin{eqnarray*}
 \mathcal{E}_0= \int^1_0 \int^{\infty}_0 \int^{\infty}_{-\infty} h^2 drdwdF.
\end{eqnarray*}
Then using the generalized Vlasov equation and integrating by parts:
\begin{eqnarray*}
 \dot{\mathcal{E}}_0&=& \int^1_0 \int^{\infty}_0 \int^{\infty}_{-\infty}[-\alpha \partial_r (h^2)+\beta \partial_w (h^2)+2hR] drdwdF\\
 &=& \int^1_0 \int^{\infty}_0 \int^{\infty}_{-\infty}[(\alpha'-\beta_w) h^2+2hR] drdwdF.
\end{eqnarray*}
Using the estimates of the last section i.e.:
\begin{eqnarray}\label{first}
 \alpha'&=&(\mu'-\lambda')\alpha=O(t^{-3}),\\
 \label{first2}\beta_w&=&\dot{\lambda}+e^{\mu-\lambda}\mu'wV^{-\frac12}=t^{-1}(1+O(t^{-2})),
\end{eqnarray}
we obtain:
\begin{eqnarray}\label{GV}
 \dot{\mathcal{E}}_0\leq t^{-1}(-1+Ct^{-2})\mathcal{E}_0+\int^1_0 \int^{\infty}_0 \int^{\infty}_{-\infty} 2hR drdwdF.
\end{eqnarray}

\subsection{Energies of higher order}
Now let us consider higher derivatives and use the following notation for derivatives of $f$:
\begin{eqnarray}\label{notation}
 \frac{\partial^{n+m}f}{\partial r^n \partial w^m}=f_{n,m},
\end{eqnarray}
notation we will also use for the derivatives of $\alpha$ and $\beta$. Define
\begin{eqnarray}\label{sum}
 \mathcal{E}_l= \sum\limits_{n+m\leq l}t^{-2m} \int^1_0 \int^{\infty}_0 \int^{\infty}_{-\infty} f_{n,m}^2 drdwdF.
\end{eqnarray}

\begin{lem}\label{VVV}
 Consider a solution to the Einstein-Vlasov system with a positive cosmological constant $\Lambda$ and surface symmetry and fix $t_0\in(0,\infty)$ for $K\leq 0$ 
 and fix $t_0\in(\Lambda^{-\frac12},\infty)$ for $K=1$,  then there is a positive constant $C_l$, depending on the solution, $l$ and $t_0$, 
 such that for all $l\in\{0\}\cup{\mathbb{N}}$
 \begin{eqnarray*}
  \frac{d\mathcal{E}_l}{dt}\leq t^{-1}(-1+C_lt^{-1})\mathcal{E}_l,
\end{eqnarray*}
for $t \geq t_0$. In particular $t\mathcal{E}_l$ is bounded to the future.
\end{lem}
\begin{proof}
 Given a solution $f$ of the Vlasov equation, the function $f_{n,m}$ satisfies the following equation
\begin{eqnarray}\label{GW}
 (\partial_t   +\alpha \partial_{r} - \beta \partial_{w})f_{n,m} =[\alpha \partial_r, \partial^n_r \partial^m _w]f - [\beta \partial_w, \partial^n_r \partial^m _w]f.
\end{eqnarray}
Thus the function $f_{n,m}$ satisfies the generalized Vlasov equation (\ref{General}) beeing $R$ the right hand side of (\ref{GW}). Let us consider one term of the sum (\ref{sum}):
\begin{eqnarray*}
 \mathcal{E}=t^{-2m}e_{n,m}=t^{-2m}\int^1_0 \int^{\infty}_0 \int^{\infty}_{-\infty} f_{n,m}^2 drdwdF,
\end{eqnarray*}
where the integral has been denoted by $e_{n,m}$. We can compute the derivative using inequality (\ref{GV}) for $e_{n,m}$ with $h$ replaced by $f_{n,m}$ and
where $R$ is now the right hand side of (\ref{GW}), obtaining:
\begin{eqnarray}\label{bust}
 \dot{\mathcal{E}}\leq (-2m-1+Ct^{-2})t^{-1} \mathcal{E} +2t^{-2m} \int^1_0 \int^{\infty}_0 \int^{\infty}_{-\infty} f_{n,m}\{[\alpha \partial_r, \partial^n_r \partial^m _w]f - [\beta \partial_w, \partial^n_r \partial^m _w]f\}drdwdF.
\end{eqnarray}
Using the notation introduced in the beginning of this section, we have due to
(\ref{first})-(\ref{first2}) and looking at the definitions of $\alpha$ and $\beta$ (\ref{alp})-(\ref{be}):
\begin{eqnarray*}
 &&\alpha_{1,0}=O(t^{-3}),\\
 &&\beta_{0,1}=t^{-1}(1+O(t^{-2})),\\
 &&\beta_{1,0}=O(t^{-4}),\\
 &&\alpha_{0,1}=e^{\mu-\lambda}V^{-2}(V-V^{-\frac12}w^2)=O(t^{-2}),\\
 &&\beta_{0,2}=e^{\mu-\lambda}\mu'V^{-\frac12}(1-w^2 V^{-1})=O(t^{-5}).
\end{eqnarray*}
Having a look at the last two expressions and having in mind Lemma \ref{fff}, we conclude that higher derivatives of these expression will decay at least at the same rate.

Consider now the integral of the right hand side of (\ref{bust}) for the case of $n+m=1$, i.e. the sum of $f_{1,0}$ and $f_{0,1}$. The integrand looks as follows:
\begin{eqnarray}\label{tir}
 2[-\alpha_{1,0}f_{1,0}^2+(\beta_{1,0}-\alpha_{0,1})t^{-1}f_{1,0}f_{0,1}+2t^{-2}\beta_{0,1}f_{0,1}^2].
\end{eqnarray}
The derivatives with respect to $r$ of $f$ are bounded and the terms of (\ref{sum}) are all positive definite. Thus one can bound $f_{1,0}$ which appears in the second term
of (\ref{tir})  with $f$ and integrating by parts the second term can be bounded by $\mathcal{E}_0$. The last term will contribute in first order with $2t^{-1}$. By an inductive argument one can treat the
higher terms. At each order a term $2t^{-1}$ will arise which cancels with the term $-2mt^{-1}$ on the right hand side of (\ref{bust}) and the lemma follows.
\end{proof}
From this lemma we prove an important consequence (compare with Lemma 63 of \cite{HakanHans}). For this purpose, we introduce the following notation:
\begin{eqnarray}\label{notation2}
 \hat{f}(t,r,w,F)=f(t,r,t^{-1}w,F).
\end{eqnarray}
\begin{thm}\label{support}
 Consider a solution to the Einstein-Vlasov system with a positive cosmological constant $\Lambda$ and surface symmetry and existence interval $(t_0,\infty)$ with $t_0\geq 0$ for 
 $K\leq 0$ and $t_0\geq \Lambda^{-\frac12}$ for $K=1$. Fix $l\in\{0\}\cup{\mathbb{N}}$. Assume that in order for 
 $(t,r,w,F)\in [t_1,\infty) \times \mathbb{S} \times \mathbb{R} \times (0, \infty)$ to be in the support of $\hat{f}$, 
 $w$ and $F$ have to satisfy $w^2+F\leq C$. Then then there is a positive constant $C_l$, depending on the solution and $l$, such that
 \begin{eqnarray*}
  \left\Vert\partial_t \hat{f}(t,\cdot)\right\Vert_{C^l[\mathbb{S}\times\mathbb{R} \times (0, \infty)]} \leq C_l t^{-3},
 \end{eqnarray*}
 for all $t\geq t_1$. Moreover a smooth, non-negative function $\hat{f}_{\infty}$ with compact support on $\mathbb{R}^{2} \times (0, \infty)$ such that
 \begin{eqnarray*}
  \left\Vert \hat{f}(t,\cdot)-\hat{f}_{\infty}\right\Vert_{C^l[\mathbb{S}\times\mathbb{R} \times (0, \infty)]} \leq C_l t^{-2},
 \end{eqnarray*}
for all $t\geq t_1$.
\end{thm}
\begin{proof}
 Having in mind the notation (\ref{notation}) and (\ref{notation2}):
  \begin{eqnarray*}
   \Upsilon = \int^1_0 \int^{\infty}_0 \int^{\infty}_{-\infty} \vert\hat{f}_{n,m}(t,r,w,F) \vert^2 drdwdF=t^{-2m}\int^1_0 \int^{\infty}_0 \int^{\infty}_{-\infty} \vert{f}_{n,m}(t,r,wt^{-1},F) \vert^2 drdwdF.
  \end{eqnarray*}
Changing variables we obtain
\begin{eqnarray*}
 \Upsilon = t^{-2m+1} \int^1_0 \int^{\infty}_0 \int^{\infty}_{-\infty} \vert{f}_{n,m}(t,r,w,F) \vert^2 drdwdF,
\end{eqnarray*}
which has the consequence due to Lemma \ref{VVV} that
\begin{eqnarray}\label{step1}
 \Upsilon \leq C E_l,
\end{eqnarray}
where $n+m\leq l$. By definition of $\hat{f}$: 
\begin{eqnarray*}
 \partial_t \hat{f}(t,r,w,F)=(\partial_t f)(t,r,t^{-1}w,F)-t^{-2}w \partial_w f(t,r,t^{-1}w,F),
\end{eqnarray*}
and using the Vlasov equation for $f$:
\begin{eqnarray*}
 \partial_t \hat{f}(t,r,w,F)=-\frac{we^{\mu-\lambda}}{\sqrt{t^2+w^2+F}}\partial_r\hat{f}(t,r,w,F)+[(\dot{\lambda}-t^{-1})w+e^{\mu-\lambda}\mu'\sqrt{t^2+w^2+F}]\partial_w \hat{f}(t,r,w,F).
\end{eqnarray*}
The terms before the partial derivatives are of order $O(t^{-4})$ and $O(t^{-3})$ respectively. Together with (\ref{step1}) and Lemma \ref{fff} this implies
that
\begin{eqnarray*}
  \left\Vert \partial_t \hat{f}(t,\cdot)\right\Vert_{C^l[\mathbb{S}\times\mathbb{R} \times (0, \infty)]} \leq C_l t^{-3},
\end{eqnarray*}
and the theorem follows.
\end{proof}

\section{Stability of surface symmetric solutions}

In the following theorem, the subscript "$bg$" refers to the background solution. For some intuition we refer to Section 7.6 of \cite{Hans}.
\begin{thm}
 Consider a surface symmetric solution to the Einstein-Vlasov system with a positive cosmological constant $\Lambda$. Choose coordinates such that the corresponding metric 
 takes the form (\ref{metric}) on $I \times\mathbb{S} \times S_K$, where $I=(t_0,\infty)$.  Choose a $t\in I$ and let $i: \mathbb{S} \times S_K \rightarrow I \times \mathbb{S} \times S_K$ be given by 
 $i(\bar{x})=(t,\bar{x})$. Let $\bar{g}_{bg}=i^*g$ and let $\bar{k}_{bg}$ denote the pullback of the second fundamental form induced on $i(\mathbb{S} \times S_K)$ by $g$. Let 
 \begin{eqnarray*}
  \bar{f}_{bg}=i^*(f \circ \pr^{-1}_{i(\mathbb{S} \times S_K)}).
 \end{eqnarray*}
Make a choice of $z$, a choice of norms as in Definition \ref{norm} and a choice of Sobolev norms on tensorfields on $\mathbb{S} \times S_K$. Then there is an $\epsilon>0$ such that 
if $(\mathbb{S} \times S_K,\bar{g},\bar{k},\bar{f})$ are initial data for (\ref{Einstein})-(\ref{Vlasov}) with $\bar{f}\in \mathcal{F}^{\infty}_z(T\Sigma)$, satisfying
\begin{eqnarray*}
 \Vert \bar{g}-\bar{g}_{bg} \Vert_{H^5} + \Vert \bar{k}-\bar{k}_{bg} \Vert_{H^5}+ \Vert \bar{f}-\bar{f}_{bg} \Vert_{H^4_z} \leq \epsilon,
\end{eqnarray*}
then the maximal globally hyperbolic development of the initial data is future causally geodesically complete.
\end{thm}
\begin{proof}
The idea of the proof consists in showing that for large times for every $\bar{x}\in{\mathbb{S}\times S_K}$, a neighbourhood of $\bar{x}$ exists such that assumptions of 
Theorem 7.16 of \cite{Hans} are satisfied and using Cauchy stability, one can conclude that there is an $\epsilon$ with the properties stated in the theorem, cf. Corollary 24.10
of \cite{Hans}.
 Fixing $n$, there is an $\epsilon$ and a constant $C_n$ such that for every $\bar{x}\in{\mathbb{S}\times S_K}$, there are geodesic normal coordinates $\bar{\mathrm{x}}$ 
 on $U=B_{\epsilon}(\bar{x})$ with respect to $\bar{g}_{\infty}$, where distances on $\mathbb{S}\times S_K $ are measured using the topological metric induced by $\bar{g}_{\infty}$.
 Denote by $\bar{g}_{\infty,ij}$ and $\bar{g}^{ij}_{\infty}$ the components of $\bar{g}_{\infty}$ and its inverse with respect to $\bar{\mathrm{x}}$. The derivatives
 of these quantities up to order $n$ with respect to $\bar{\mathrm{x}}$ on $U$ are bounded by $C_n$. The same is true for the derivatives of $\bar{\mathrm{x}}$ and its inverse
 considered as a function of $(r,\theta,\varphi)$. These uniform bounds hold regardless of the base point, cf. Lemma 34.9 of \cite{Hans}. Define $L$ by $e^L=4/H$ 
 and define the coordinates $\bar{\mathrm{y}}=e^{-L}t\bar{\mathrm{x}}$. The range of $\bar{\mathrm{y}}$ is $B_{e^{-L}t\epsilon}(0)$. From now on we assume that $t$ is large
 enough such that $e^{-L}t\epsilon>1$, and we assume the coordinates of $\bar{\mathrm{y}}$ to be defined on the image if $B_1(0)$ under $\bar{\mathrm{y}}^{-1}$. Let $\bar{\mathrm{g}}$ 
 denote the components of $\bar{g}(t,\cdot)$ with respect to the coordinates $\bar{\mathrm{y}}$ and $\bar{\mathrm{g}}_{\infty,ij}$ the components of $\bar{{g}}_{\infty}$ with respect
 to the coordinates $\mathrm{x}$. Due to (\ref{guess}), we have:
 \begin{eqnarray*}
 \vert \partial^{\alpha}(e^{-2L}\bar{\mathrm{g}}_{ij}-\bar{\mathrm{g}}_{\infty,ij}) \circ \bar{\mathrm{y}}^{-1}\vert \leq C_n t^{-3-\vert\alpha \vert},
 \end{eqnarray*}
on $B_1(0)$ for $\vert \alpha \vert\leq n$. By definition:
 \begin{eqnarray*}
  \bar{\mathrm{g}}_{\infty,ij} \circ \bar{\mathrm{y}}^{-1}(0)=\delta_{ij},
 \end{eqnarray*}
and  
 \begin{eqnarray*}
  \bar{y}^{-1}(\bar{\zeta})=\bar{x}^{-1}(e^L t^{-1}\bar{\zeta}).
 \end{eqnarray*}
Thus,
\begin{eqnarray*}
 \vert ( \partial_l \bar{\mathrm{g}}_{\infty,ij}) \circ \bar{\mathrm{y}}^{-1}(u) \vert = \left\vert \sum_{m=1}^n\int^1_0 \frac{d}{ds} [(\partial_m \partial_l\bar{\mathrm{g}}_{\infty,ij})\circ \bar{\mathrm{y}}^{-1}(su)u^m] ds\right\vert \leq C e^{-H\tau},
\end{eqnarray*}
and
\begin{eqnarray*}
 \vert e^{-2L}\bar{\mathrm{g}}_{ij}\circ \bar{\mathrm{y}}^{-1}-\delta_{ij} \vert \leq C e^{-3H\tau}.
\end{eqnarray*}
To conclude, we thus have:
\begin{eqnarray*}
 \Vert e^{-2L} \partial_l \bar{\mathrm{g}}_{ij} \circ \bar{\mathrm{y}}^{-1} \Vert_{C^{n-1}[B_1(0)]} \leq C_n e^{-H\tau},
\end{eqnarray*}
and 
\begin{eqnarray*}
 \Vert (\bar{\mathrm{k}}_{ij}- H \bar{\mathrm{g}}_{ij}) \circ \bar{\mathrm{y}}^{-1} \Vert_{C^{n-1}[B_1(0)]} \leq C_n e^{-2H\tau},
\end{eqnarray*}
where $\bar{\mathrm{k}}_{ij}$ denote the components of $\bar{k}$ with respect to the coordinates $\bar{\mathrm{y}}$. Turning to the distribution function what remains to be estimated
(cf. (7.34) and (7.37) of \cite{Hans}) is 
\begin{eqnarray*}
 \Pi=\sum_{\vert \alpha \vert +b\leq k_0 }\int^{\infty}_{-\infty}\int^{\infty}_{0} \int_{\bar{y}(U)} e^{-2b\omega} (1+e^{2\omega}w^2+e^{2\omega}t^{-2}F)^{z+b}\vert\partial^{\alpha}_{\bar{\zeta}}\partial_{w}^b \bar{\mathrm{f}}_{\bar{\mathrm{y}}} \vert^2(\bar{\zeta},w,F) d\bar{\zeta}dFdw,
\end{eqnarray*}
where
\begin{eqnarray*}
 \bar{\mathrm{f}}_{\bar{\mathrm{y}}}=\bar{f}[\bar{\mathrm{y}}^{-1}(\bar{\zeta}),w,F]=\bar{f}[\bar{\mathrm{x}}^{-1}(e^Lt^{-1}\bar{\zeta}),w,F].
\end{eqnarray*}
Now following exactly the same arguments as in \cite{HakanHans}, using our Theorem \ref{support},
choosing $\omega=K+K_{\mathrm{Vl}}$ with $K_{\mathrm{Vl}}=\ln t$, $k_0=4$, $z>\frac52$ freely, we obtain that
\begin{eqnarray*}
 \Pi \leq C_{z,\alpha,b}t^{-2\vert \alpha \vert-3}\leq H^2 \epsilon^{\frac52}e^{-3K/2}t^{-1},
\end{eqnarray*}
for large $t$ large enough, where the last inequality is the one we need in order to be able to apply Theorem 7.16 of \cite{Hans}. Since we are considering the case of a cosmological
constant the function $V$ in the theorem can be set equal to $\Lambda(1+\phi^2)$.
\end{proof}
In addition to the conclusions of the theorem one obtains via Theorem 7.16 of \cite{Hans} a detailed description of the asymptotics. In particular the cosmic no hair theorem holds
for the perturbed solution.

\section{Outlook}
We have treated the special case of a vanishing scalar field of Theorem 7.16 of \cite{Hans}. An obvious generalization is thus the extension to a non-vanishing scalar field, 
a starting point could be \cite{Svedberg,Tegan} together with the results of \cite{Hans}. Another natural question is to consider the Einstein-Maxwell system, i.e. charged particles,
cf. for results in this direction \cite{Luo, Svedberg,Tchap3}. It is also of interest to consider the Einstein-Euler system, where shocks may play a role, cf. \cite{LeFloch3,LeFloch2,Tchap4}.

\section{Acknowledgements}
The author thanks Hans Ringstr\"{o}m for suggesting this problem, for several helpful discussions, reading parts of the manuscript of this work 
and the opportunity to read the manuscripts of \cite{HakanHans,Hans}.
He is grateful to the G\"{o}ran Gustafsson Foundation for Research in Natural Sciences and Medicine for their financial support. This material is partly based 
upon work supported by the National Science Foundation under Grant No. 0932078 000, 
while the author was in residence at the Mathematical Sciences Research Institute in 
Berkeley, California, during October 2013. He is also grateful for the invitation of the Institute of Theoretical Physics, Charles University Prague, where the author spent part of November 2013. He is currently supported by the Irish Research Council. 
\bibliographystyle{abbrv}

\end{document}